\documentclass[3p,preprint]{elsarticle}

\journal{Theoretical Computer Science}

\usepackage{amssymb}

\usepackage[figuresright]{rotating}

\newtheorem{theorem}{Theorem}
\newtheorem{lemma}[theorem]{Lemma}
\newtheorem{proposition}[theorem]{Proposition}
\newtheorem{corollary}[theorem]{Corollary}

\newtheorem{example}{Example}

\newproof{proof}{Proof}

\usepackage[utf8]{inputenc}
\usepackage[algosection]{algorithm2e}
\usepackage{endnotes}
\usepackage{listings}
\usepackage{mathtools}
\usepackage{verbatim}
\usepackage{mathrsfs}
\usepackage{longtable}
\usepackage{enumitem}
\usepackage{etoolbox}
\usepackage{float}
\usepackage{color}
\usepackage{hyperref}

\usepackage{complexity}

\newclass{\COMSLIP}{COM\mbox{-}SLIP}
\newclass{\COMSLIPCUP}{COM\mbox{-}SLIP^{\cup}}

\newclass{\DCM}{DCM}
\newclass{\eDCM}{eDCM}
\newclass{\eNPDA}{eNPDA}
\newclass{\DPDA}{DPDA}
\newclass{\RDPDA}{RDPDA}
\newclass{\PDA}{PDA}
\newclass{\DCMNE}{DCM_{NE}}
\newclass{\TwoDCM}{2DCM}
\newclass{\NCM}{NCM}
\newclass{\eNCM}{eNCM}
\newclass{\eNQA}{eNQA}
\newclass{\eNSA}{eNSA}
\newclass{\eNPCM}{eNPCM}
\newclass{\eNQCM}{eNQCM}
\newclass{\eNSCM}{eNSCM}
\newclass{\DPCM}{DPCM}
\newclass{\NPCM}{NPCM}
\newclass{\NQCM}{NQCM}
\newclass{\NSCM}{NSCM}
\newclass{\NPDA}{NPDA}
\newclass{\TRE}{TRE}
\newclass{\NFA}{NFA}
\newclass{\DFA}{DFA}
\newclass{\NCA}{NCA}
\newclass{\DCA}{DCA}
\newclass{\DTM}{DTM}
\newclass{\NTM}{NTM}
\newclass{\NTMCM}{NTMCM}
\newclass{\DLOG}{DLOG}
\newclass{\CFG}{CFG}
\newclass{\ETOL}{ET0L}
\newclass{\EDTOL}{EDT0L}
\newclass{\CFP}{CFP}
\newclass{\ORDER}{O}
\newclass{\MATRIX}{M}
\newclass{\BD}{BD}
\newclass{\LB}{LB}
\newclass{\ALL}{ALL}
\newclass{\decLBD}{decLBD}
\newclass{\StLB}{StLB}
\newclass{\SBD}{SBD}
\newclass{\TCA}{TCA}
\newclass{\LL}{{\cal L}}

\newclass{\CSA}{CSA}

\newclass{\DCSA}{DCSA}
\newclass{\NCSA}{NCSA}
\newclass{\GSM}{GSM}
\newclass{\RDCSA}{RDCSA}
\newclass{\RNCSA}{RNCSA}
\newclass{\DCSACM}{DCSACM}
\newclass{\NCSACM}{NCSACM}

\DeclareMathOperator{\push}{push}
\DeclareMathOperator{\pop}{pop}
\DeclareMathOperator{\stay}{stay}

\begin{document}

\begin{frontmatter}

\title{Variations of Checking Stack Automata:  Obtaining Unexpected Decidability Properties
\tnoteref{t1}\tnoteref{t2}}

\tnotetext[t1]{A preliminary version of this paper has appeared in the 
Springer LNCS Proceedings of
the 21st International Conference on Developments in Language Theory (DLT 2017), pp. 235--246.}

\tnotetext[t2]{\textcopyright 2018. This manuscript version is made available under the CC-BY-NC-ND 4.0 license \url{http://creativecommons.org/licenses/by-nc-nd/4.0/}}

\author[label1]{Oscar H. Ibarra\fnref{fn1}}
\address[label1]{Department of Computer Science\\ University of California, Santa Barbara, CA 93106, USA}
\ead[label1]{ibarra@cs.ucsb.edu}
\fntext[fn1]{Supported, in part, by
NSF Grant CCF-1117708 (Oscar H. Ibarra).}

\author[label2]{Ian McQuillan\fnref{fn2}}
\address[label2]{Department of Computer Science, University of Saskatchewan\\
Saskatoon, SK S7N 5A9, Canada}
\ead[label2]{mcquillan@cs.usask.ca}
\fntext[fn2]{Supported, in part, by Natural Sciences and Engineering Research Council of Canada Grant 2016-06172 (Ian McQuillan).}

\begin{abstract}
We introduce a 
model of one-way language acceptors 
(a variant of a checking stack automaton) and 
show the following decidability properties:
\begin{enumerate}
\item The deterministic version has a decidable membership problem
but has an undecidable emptiness problem.
\item The nondeterministic version has an undecidable membership 
problem and emptiness problem.
\end{enumerate}
There are many models of accepting devices
for which there is no difference with these problems between deterministic and
nondeterministic versions, i.e., the membership problem for
both versions are either decidable or undecidable, and the
same holds for the emptiness problem.
As far as we know, the model we introduce above is the first
one-way model
to exhibit properties (1) and (2).
We define another family of one-way acceptors where the nondeterministic
version has an undecidable emptiness problem,
but the deterministic version has a decidable emptiness problem.
We also know of no other model with this property in the literature. 
We also investigate decidability properties of
other variations of checking stack automata (e.g.,
allowing multiple stacks, two-way input, etc.).
Surprisingly, two-way deterministic machines with multiple checking stacks
and multiple reversal-bounded counters are shown to have a decidable membership problem, a very general model with this property.
\end{abstract}

\begin{keyword}
checking stack automata \sep pushdown automata \sep decidability \sep reversal-bounded counters
\end{keyword}

\end{frontmatter}

\section{Introduction} \label{sec:intro}

The deterministic and nondeterministic versions of
most known models of language acceptors exhibit the same
decidability properties for each of the membership and emptiness problems. 
In fact, it is possible to define machine models in a general fashion by varying the allowed store types, such as with
 Abstract Families of Acceptors (AFAs) from
\cite{G75}, or a similar type of machine model with abstract store types
used in \cite{StoreLanguages} and in this paper. 
Studying machine models defined in such a general fashion is advantageous as certain decidability problems are equivalently decidable for arbitrary machine models defined using such a framework, and therefore  it is possible to see which problems could conceivably differ in terms of decidability.
For arbitrary one-way machine models defined in the way used here, the emptiness problem for the nondeterministic machines of this class, the membership problem for nondeterministic machines of this class, and the emptiness problem for the deterministic machines in this class, must all be either decidable or undecidable. Membership for deterministic machines could conceivably differ from the other three decidability problems. However, as far as we know, no one-way model has been shown to
exhibit different decidability properties for deterministic
and nondeterministic versions.  The question arises of whether there is a model where membership for deterministic machines is decidable while it is undecidable for nondeterministic machines?

A second topic of interest here is that of studying decidability properties of different classes of machines when adding additional data stores. In \cite{Harju2002278}, it was shown that for any one-way machine model (defined using another method used there), if the languages accepted by these machines are all semilinear\footnote{See \cite{harrison1978} for the formal definition. Equivalently, a language is semilinear if it has the same commutative closure as some regular language.}, then augmenting these machines with additional reversal-bounded counters\footnote{A counter stores a non-negative integer that can be tested for zero, and it is reversal-bounded if there is a bound on the number of changes between non-decreasing and non-increasing.} produces only semilinear languages. And, if semilinearity is effective with the original model, then it is also effective after adding counters, and therefore the emptiness problem is decidable. However, it is unknown what can occur when augmenting a class of machines that accepts non-semilinear languages with reversal-bounded counters. Can adding such counters change decidability properties? 

These two topics are both simultaneously studied in this paper. Of primary importance  is the one-way checking stack automaton \cite{CheckingStack}, which is similar to a pushdown
automaton that cannot erase its stack, but can enter and read the
stack in two-way read-only mode, but once this mode is entered,
the stack cannot change. This model accepts non-semilinear languages, but has decidable emptiness and membership problems. Here, we introduce a new
model of one-way language acceptors by
augmenting a checking stack automaton with
reversal-bounded counters, and the deterministic and
nondeterministic versions are denoted by 
$\DCSACM$ and $\NCSACM$, respectively.
The models with two-way input (with end-markers) are called
$2\DCSACM$ and $2\NCSACM$. These are
generalized further to models with $k$ checking stacks:
$k$-stack $2\DCSACM$ and
$k$-stack $2\NCSACM$.
These models can be defined within the general machine model framework mentioned above.

We show the following results concerning membership and
emptiness:

\begin{enumerate}
\item
The membership and emptiness problems for $\NCSACM$s are
undecidable, even when there are only two reversal-bounded counters.
\item
The emptiness problem for $\DCSACM$ is decidable when there
is only one reversal-bounded counter but undecidable when 
there are two reversal-bounded counters.
\item
The membership problem for $k$-stack $2\DCSACM$s is decidable
for any $k$.
\end{enumerate}
Therefore, this machine model provides the first known model where membership is decidable for deterministic machines, while the other decidability properties are undecidable, which is the only property that can conceivably differ. It also shows one possible scenario that can occur when augmenting a machine model accepting non-semilinear languages with reversal-bounded counters: it can change the emptiness problem for both nondeterministic and deterministic machines to be undecidable, as with the membership problem for nondeterministic machines, but membership for deterministic machines can remain decidable (and therefore, all such languages accepted by deterministic machines are recursive).

In addition, we define another family of one-way acceptors where the deterministic
version has a decidable emptiness problem, but the nondeterministic version
has an undecidable emptiness problem.
This model must necessarily not be defined using the general machine model framework, as emptiness for deterministic and nondeterministic machine models are always equivalently decidable. But the model is still natural and demonstrates what must occur to obtain unusual decidable properties.
Further, we introduce a new family with decidable
emptiness, containment, and equivalence problems, which is one of
the most powerful families to have these properties (one-way
deterministic machines with one reversal-bounded counter and a 
checking stack that can only read from the stack at the end of the
input).
We also investigate the decidability properties of
other variations of checking stack automata (e.g.,
allowing multiple stacks, two-way input, etc.).

\section{Preliminaries}

This paper requires basic knowledge of automata and
formal languages, including finite automata, pushdown automata, and Turing machines \cite{HU}. An alphabet $\Sigma$ is a (usually finite unless stated otherwise) set of symbols. The set $\Sigma^*$ is the set of all words over $\Sigma$, which contains the empty word $\lambda$. A language is any set $L \subseteq \Sigma^*$. Given a word $w \in \Sigma^*$, $|w|$ is the length of $w$. A language $L$ is bounded if there exists words $w_1, \ldots, w_k$ such that $L \subseteq w_1^* \cdots w_k ^*$, and $L$ is letter-bounded if $w_1, \ldots, w_k$ are letters.

We use a variety of machine models here, mostly
built on top of the checking stack. It is possible
to define each machine model directly. 
As discussed in Section \ref{sec:intro},
an alternate
approach is 
to define ``store types'' first, which describes just the
behavior of the store, including instructions that can change the
store, and the manner in which the store can be read.
This can capture standard types of stores studied in the
literature, such as a pushdown, or a counter.
Defined generally enough, it can also define a checking
stack, or a reversal-bounded counter.
Then, machines using one or more store types can be defined,
in a standard fashion. A $(\Omega_1, \ldots, \Omega_k)$-machine
is a machine with $k$ stores, where $\Omega_i$ describes each store.
This is the approach taken here, in a similar fashion to the one taken in 
\cite{EngelfrietCheckingStack} or \cite{G75} to define these same types of automata.
This generality will also help in illustrating what is required to obtain certain decidability properties; see e.g.\
Lemma \ref{generallambda} and Proposition \ref{equivalentdecidability} which are proven generally for arbitrary
store types. Furthermore, these two results are used many times within other proofs rather than having many ad hoc proofs. Hence, this generality in defining machines serves several purposes for this work.

First, store types, and machines using store types are defined
formally using the same framework used by the authors in \cite{StoreLanguages}. 
A store type is a tuple
$\Omega = (\Gamma, I,f,g,c_0, L_I)$,
where $\Gamma$ is the store alphabet (potentially infinite available to all
machines using this type of store), $I$ is the 
set of allowable instructions, $c_0$ is the initial
configuration which is a word in $\Gamma^*$, and
$L_I \subseteq I^*$ is the instruction language (over possibly an infinite
alphabet) of allowable sequences
of instructions, $f$ is the read function, a partial
function from $\Gamma^*$ to $\Gamma$, and $g$ is the write
function, a partial function from 
$\Gamma^* \times I$ to $\Gamma^*$.

We will study a few examples of store types.
First, a pushdown store type 
is a tuple 
$\Omega = (\Gamma, I,f,g,c_0, L_I)$,
where 
$\Gamma$ is an infinite set of store
symbols available to pushdowns, where special symbol $Z_b\in \Gamma$ is the bottom-of-stack marker, and 
$\Gamma_0 = \Gamma - \{Z_b\}$,
$I = \{\push(y) \mid y \in \Gamma_0\} \cup \{\pop,\stay\}$
is the set of instructions of the pushdown, where
the first set are called the push instructions, and the second set
contains the pop and stay instruction, 
$L_I = I^*$, $c_0 = Z_b $,
$f( x a) = a, a \in \Gamma, x\in \Gamma^*$ with 
$xa \in Z_b \Gamma_0^*$, and
$g$ is defined as:
\begin{itemize}
\item $g(Z_b x, \push(y)) = Z_b x y$ for $x \in \Gamma_0^*, y \in \Gamma_0$,
\item $g(Z_b x a, \pop) = Z_b x$, for $x \in \Gamma_0^*, a \in \Gamma_0$,
\item $g(Z_b x , \stay) = Z_b x $, for $x \in \Gamma_0^*$.
\end{itemize}

A counter store tape
can be obtained by restricting the pushdown store type to only having
a single symbol $c \in \Gamma_0$ (plus the bottom-of-stack marker). 
The instruction language $L_I$ in the definition of $\Omega$ restricts the allowable sequences of instructions available to the store type $\Omega$ (that is, a computation can only accept if its sequence of instructions is in the instruction language). This restriction does not exist in the definition of AFAs, but can be used to define many classically studied machine models, while still preserving many useful properties.
For example, an $l$-reversal-bounded counter store type
is a counter store type
with $L_I$ equal to the alternating concatenation of 
$\{\push(c),\stay\}^*$ and $\{\pop , \stay\}^*$ with $l$
applications of concatenation (this is more classically stated as,
there are at most $l$ alternations between non-decreasing and non-increasing).

Next, the more complicated stack store type is a tuple 
$\Omega = (\Gamma, I,f,g,c_0, L_I)$,
where 
\begin{itemize}
\item $\Gamma$ is an infinite set of store
symbols available to stacks, where special symbols $\downarrow \in \Gamma$
are the position of the read/write head in the stack, $Z_b\in \Gamma$ is the bottom-of-stack marker, and $Z_t \in \Gamma$ is the top-of-stack marker, with
$\Gamma_0 = \Gamma - \{\downarrow, Z_b, Z_t\}$,
\item $I = \{\push(y) \mid y \in \Gamma_0\} \cup \{\pop,\stay\}\cup \{{\rm D}, {\rm S}, {\rm U}\}$
is the set of instructions of the stack, where
the first set are called the push instructions, the second set
is the pop and stay instruction, and the third set are the move instructions (down, stay, or up),
\item $L_I = I^*$, $c_0 = Z_b \downarrow  Z_t$,
$f( x a \downarrow x') = a, 
a \in \Gamma_0 \cup \{Z_t, Z_b\}, x,x'\in \Gamma^*$ with 
$xax' \in Z_b \Gamma_0^* Z_t$,
\item and $g$ is defined as:
\begin{itemize}
\item $g(Z_b x \downarrow Z_t, \push(y)) = Z_b x y \downarrow Z_t$ for $x \in \Gamma_0^*, y \in \Gamma_0$,
\item $g(Z_b x a \downarrow Z_t, \pop) = Z_b x \downarrow Z_t$, for $x \in \Gamma_0^*, a \in \Gamma_0$,
\item $g(Z_b x \downarrow Z_t, \stay) = Z_b x \downarrow Z_t$, for $x \in \Gamma_0^*$,
\item $g(Z_b x a \downarrow x'   , {\rm D}) = Z_b x \downarrow a x'$,
for $x, x' \in \Gamma^*, a \in \Gamma_0 \cup \{Z_t\}$, with $xax' \in \Gamma_0^* Z_t$,
\item $g(Z_b x \downarrow x', {\rm S}) = Z_b x\downarrow x'$, for
$x,x' \in \Gamma^*, x x' \in \Gamma_0^* Z_t$,
\item $g(Z_b x \downarrow a x', {\rm U}) = Z_b x a \downarrow x'$,
for $x,x' \in \Gamma^*, a \in \Gamma_0 \cup \{Z_t\}, xax' \in \Gamma_0^* Z_t$.
\end{itemize}

\end{itemize}
That is, a stack is just like a pushdown with the additional ability to read from the ``inside'' of the stack (but not change the inside) in two-way read-only mode.
Also, the checking stack store type is a restriction of stack store type above where $L_I$ is restricted to be in 
 $\{\push(y), \stay \mid y \in \Gamma_0\}^* \{{\rm D}, {\rm S}, {\rm U}\}^*$. 
 That is, a checking stack has two phases, a ``writing phase'',
 where it can push or stay (no pop), and then a ``reading phase'', where it
 enters the stack in read-only mode. But once it starts reading, it cannot change the stack again.


Given store types $(\Omega_1, \ldots, \Omega_k)$, with
$\Omega_i = (\Gamma_i, I_i,f_i,g_i,c_{0,i}, L_{I_i})$, a
two-way $r$-head $k$-tape 
$(\Omega_1, \ldots, \Omega_k)$-machine is a tuple
$M = (Q,\Sigma, \Gamma, \delta,  \rhd,\lhd, q_0, F)$
where $Q$ is the finite set of states, $q_0 \in Q$ is the
initial state, $F \subseteq Q$ is the set of final states,
$\Sigma$ is the finite input alphabet, 
$\Gamma$ is a finite subset of the store alphabets
of $\Gamma_1 \cup \cdots \cup \Gamma_k$,
$\delta$
is the finite transition relation from
$Q\times [\Sigma]^r \times \Gamma_1 \times \cdots \times \Gamma_k$ 
to
$Q \times I_1 \times \cdots \times I_k \times [\{-1,0,+1\}]^r$.

An instantaneous description (ID) is a tuple
$(q,\rhd w \lhd, \alpha_1, \ldots, \alpha_r, x_1, \ldots, x_k)$,
where $q \in Q$ is the current state, $w$ is the current
input word (surrounded by left input end-marker and right input
end-marker), $0 \leq \alpha_j \leq |w|+1$ is the current
position of tape head $j$ (this can be thought of as $0$ scanning 
$\rhd$, and $|w|+1$ scanning $\lhd$), for $1 \leq j \leq r$, and $x_i \in \Gamma_i^*$ is the current word in the
$\Omega_i$ store, for $1 \leq i \leq k$. Then $M$ is deterministic
if $\delta$ is a partial function (i.e.\ it only maps each element to
at most one element).

Then $(q,\rhd w \lhd, \alpha_1, \ldots, \alpha_r, x_1, \ldots, x_k)
\vdash_M (q',\rhd w \lhd, \alpha_1', \ldots, \alpha_r', x_1', \ldots, x_k')$, (two IDs)
if there exists a transition
$(q', \iota_1, \ldots, \iota_k, \gamma_1, \ldots, \gamma_r) \in 
\delta(q, a_1, \ldots, a_r, b_1, \ldots, b_k)$, where
$a_j$ is character $\alpha_j+1$ of $\rhd w \lhd$ ($1$ is added since $\rhd$ is the letter at position $0$), and $\alpha_j' = \alpha_j + \gamma_j$, for $1 \leq j \leq r$, $b_i = f_i(x_i)$,
and $g_i(x_i, \iota_i) = x_i'$ for $1 \leq i \leq k$. 
Instead of $\vdash_M$, we can also write $\vdash_M^{(\iota_1, \ldots, \iota_k)}$ to show the instructions applied to each store on the transition.
We let $\vdash_M^*$ be the reflexive
and transitive closure of $\vdash_M$, and let $\vdash_M^{(\gamma_1, \ldots, \gamma_k)}$, where
$\gamma_i \in I_i^*$ is the sequence of instructions applied to store $i$, $ 1 \leq i \leq k$, in the sequence of transitions applied, and $|\gamma_1| = \cdots = |\gamma_k|$.
The language accepted by $M$, $L(M)$, is equal to 
$$\{ w \mid  (q_0, \rhd w \lhd, 1, \ldots, 1, c_{0,1}, \ldots, c_{0,k}) \vdash_M^{(\gamma_1, \ldots, \gamma_k)} (q_f, \rhd w \lhd, \alpha_1, \ldots, \alpha_r, x_1, \ldots, x_k),q_f \in F, \gamma_i \in L_{I_i}, 1 \leq i \leq k\}.
$$
Thus, the sequence of instructions applied to each store must satisfy its instruction language, and they must each be of the same length.

The different {\em machine modes} are combinations of
either one-way or two-way, 
deterministic or nondeterministic,
and $r$-head for some $r \geq 1$.
For example, one-way, $1$-head, deterministic, is a machine mode.
Given a sequence of store types $\Omega_1, \ldots, \Omega_k$ and a machine mode, one can study the set of all $(\Omega_1, \ldots, \Omega_k)$-machines with this mode. The set of all such machines with a mode is said to be {\em complete}. Any strict subset is said to be {\em incomplete}.
Given a set of (complete or incomplete) machines ${\cal M}$ of this type, the family of languages
accepted by these machines is denoted ${\cal L}({\cal M})$.
For example, the set of all one-way deterministic pushdown automata is complete as it contains all one-way deterministic machines that use the pushdown store. 
But consider the set of all one-way deterministic pushdown automata that can only decrease the size of the stack when
scanning the right end-marker. This is a strict subset of all one-way deterministic machines that use the pushdown store, since the instructions available to such machines
depend on the location of the input (whether it has reached the end of the input or
not). Therefore, this is an incomplete set of machines. The instruction language of a store does allow a complete class of machines to restrict the allowable sequences of instructions, but it has to apply to all machines using the store.
Later in the paper, we will consider variations of checking stack automata such as one called {\em no-read}, which means that they do not read from the inside of the checking stack before hitting the right input end-marker. This is similarly an incomplete set of automata since the instructions allowed differs depending on the input position.

The class of one-way deterministic
(resp.\ nondeterministic) checking stack automata
is denoted by $\DCSA$ (resp., $\NCSA$) \cite{CheckingStack}.
The class of deterministic (resp.\ nondeterministic),
finite automata is denoted by $\DFA$ (resp., $\NFA)$ \cite{HU}.
For $k,l \geq 1$, the class of one-way deterministic (resp.\ nondeterministic)
$l$-reversal-bounded $k$-counter machines
is denoted by $\DCM(k,l)$ (resp.\ $\NCM(k,l)$). If only one
integer is used, e.g.\ $\NCM(k)$, this class contains
all $l$-reversal-bounded $k$ counter machines, for some $l$,
and if the integers is omitted, e.g., $\NCM$ and $\DCM$, they contain
all $l$-reversal-bounded $k$ counter machines, for some $k,l$.
Note that a counter that makes $l$ reversals
can be simulated by $\lceil \frac{l+1}{2} \rceil$ 1-reversal-bounded
counters \cite{Ibarra1978}. Closure and decidable properties of various
machines augmented with reversal-bounded counters have been studied
in the literature (see, e.g., \cite{Harju2002278,Ibarra1978}).
For example, it is known that the membership and emptiness
problems are decidable for $\NCM$ \cite{Ibarra1978}.

Also, here we will study the following new classes of machines
that have not been studied in the literature:
one-way deterministic (resp.\ nondeterministic) 
machines defined by stores consisting of one
checking stack and $k$ $l$-reversal-bounded
counters, denoted by
$\DCSACM(k,l)$ (resp.\ $\NCSACM(k,l)$), those
with $k$-reversal-bounded counters, denoted by $\DCSACM(k)$ 
(resp.\ $\NCSACM(k)$), and those with some number of reversal-bounded counters, denoted by $\DCSACM$ (resp.\ $\NCSACM$).

All models above also have two-way versions of the machines defined, denoted by preceding
them with 2, e.g.\
$2\DCSA, 2\NCSA, 2\NCM(1), 2\DFA, 2\NFA$, etc.
We will also define models with $k$ checking stacks for some $k$,
which we will precede with the phrase ``$k$-stack'', e.g.\
$k$-stack $2\DCSA$, $k$-stack $2\NCSA$,
$k$-stack $2\DCSACM$, $k$-stack $2\NCSACM$, etc. 
When $k=1$, then this corresponds with omitting the phrase ``$k$-stack''.



\section{A Checking Stack with Reversal-Bounded Counters}


Before studying the new types of stores and machine models,
we determine several properties that are equivalent for any
complete set of machines. This helps to demonstrate what is
required to potentially have a machine model
where the deterministic version has a decidable membership problem
with an undecidable emptiness problem, while both problems are undecidable for the nondeterministic version.

First, we examine a machine's behavior on one word.
\begin{lemma} Let $M$ be a one- or two-way, $r$-head, for some
$r\geq 1$, $(\Omega_1, \ldots, \Omega_k)$-machine, 
and let $w\in \Sigma^*$.
We can effectively construct another $(\Omega_1, \ldots, \Omega_k)$-machine $M_w$ that is one-way and $1$-head
which accepts $\lambda$ if and only if $M$ accepts $w$.
Furthermore, $M_w$ is deterministic if $M$ is deterministic.
\label{generallambda}
\end{lemma}
\begin{proof}  The input $w$ is encoded in the state of $M_w$, and
$M_w$ on input $\lambda$, simulates the computation of $M$ 
and accepts $\lambda$ if and only if $M$ accepts $w$. This uses
a subset of the sequence of transitions used by $M$ (and thereby would satisfy the instruction language of each store). Since $M_w$
is only on $\lambda$ input, two-way input is not needed in
$M_w$, and the $r$-heads are simulated completely in the finite
control.
\qed \end{proof}

Then, for all machines with the same store types, the following decidability problems are equivalent:
\begin{proposition}
\label{equivalentdecidability}
Consider store types $(\Omega_1, \ldots, \Omega_k)$.
The following problems are equivalently decidable, for the stated
complete sets of automata:
\begin{enumerate}
\item the emptiness problem for one-way deterministic 
$(\Omega_1, \ldots, \Omega_k)$-machines,
\item the emptiness problem for one-way nondeterministic 
$(\Omega_1, \ldots, \Omega_k)$-machines,
\item the membership problem for one-way nondeterministic 
$(\Omega_1, \ldots, \Omega_k)$-machines,
\item acceptance of $\lambda$, for one-way nondeterministic 
$(\Omega_1, \ldots, \Omega_k)$-machines,
\item the membership problem for two-way $r$-head (for $r \geq 1$) nondeterministic $(\Omega_1, \ldots, \Omega_k)$-machines.
\end{enumerate}
\end{proposition}
\begin{proof} 
The equivalence of 1) and 2) can be seen by taking a nondeterministic machine $M$.
Let $T = \{t_1, \ldots, t_m\}$ be labels in bijective correspondence
with the transitions of $M$.
Then construct $M'$ which
operates over alphabet $T$.
Then $M'$ reads each input symbol $t$ and simulates $t$ of $M$
on the store, while always moving right on the input. However, if it is a stay
transition on the input of $M$, then $M'$ also checks that the next input
symbol read (if any), $t'$, is defined on the same letter of $\Sigma$
in $M$.
Then $M'$ is deterministic, and changes its
stores identically in sequence to $M$ (thereby still satisfying the same instruction language), and $L(M')$ is therefore
empty if and only if $L(M)$ is empty.

It is immediate that 5) implies 4), and it follows that 4) implies 5) from
Lemma \ref{generallambda}. Similarly, 3) implies 4), and
4) implies 3) from Lemma \ref{generallambda}.

To show that 4) implies 2), notice that any complete set
of nondeterministic one-way automata are closed under
homomorphisms $h$ where $h(a) \leq 1$, for all letters $a$.
Considering the homomorphism that erases all letters,
the resulting language is empty if and only if $\lambda$ is
accepted by the original machine.

To see that 2) implies 4), take a one-way nondeterministic
machine and make a new one that cannot accept if there is an 
input letter. This new machine is non-empty if and only if
$\lambda$ is accepted in the original machine.
\qed \end{proof}

It is important to note that this proposition is not necessarily
true for incomplete sets of automata, as the machines constructed in the proof need to
be present in the set. We will see some natural restrictions later where
this is not the case, such as sets of machines where there is a
restriction on what instructions can be performed on the store based on the position of the input. And indeed, to prove the equivalence of 1) and 2)
above, the deterministic machine
created reads a letter for every transition of the nondeterministic machine applied.
Hence, consider a set of machines that is only allowed to apply a strict subset of store instructions before the end-marker. Let $M$ be a nondeterministic machine of this type, and say that $M$ applies some instruction on the end-marker that is not available to the machine before the end-marker. But
the deterministic machine $M'$ created from $M$ in Proposition \ref{equivalentdecidability} reads an input letter when every instruction is applied, even including those applied on the end-marker of $M$. But since $M'$ is reading an input letter during this operation, it would violate the instructions allowed by $M'$ before the end-marker.

The above proposition indicates that for every complete set of one-way
machines, membership for nondeterminism, emptiness
for nondeterminism, and emptiness for determinism are equivalent.
Thus, the only problem that can potentially differ is membership for
deterministic machines. Yet we know of no existing model where
it differs from the other three properties. We examine one next.

We will study $\NCSACM$s and $\DCSACM$s, which are $\NCSA$s and $\DCSA$s (nondeterministic and deterministic checking stack automata) 
respectively, augmented by reversal-bounded counters. First, two examples
will be shown, demonstrating a language that can be accepted by a $\DCSACM$.
\begin{example}
\label{DCSAwitness}
Consider the language 
$L = \{(a^n\#)^n ~|~ n \ge 1\}$.
A $\DCSACM$ $M$ with one 1-reversal-bounded counter
can accept $L$ as follows:  $M$ when given an input $w$
(we may assume that the input is of the form
$w = a^{n_1} \# \cdots a^{n_k}\#$ for some $k \ge 1$
and $n_i \ge 1$ for $1 \le i \le k$, since the finite
control can check this), copies the first
segment $a^{n_1}$ to the stack while also storing number
$n_1$ in the counter. Then $M$ goes up and down
the stack comparing $n_1$ to the rest of the input 
to check that $n_1 = \cdots = n_k$
while decrementing the counter by 1 for each
segment it processes.  Clearly, $L(M) = L$
and $M$ makes only 1 reversal on the counter.
We will show in Proposition \ref{Grei} that $L$ cannot
be accepted by an $\NCSA$ (or an $\NCM$).
\end{example}

\begin{example}
\label{2DCM1witness}
Let $L = \{a^i b^j c^k ~|~ i, j \ge 1, k = i \cdot j \}$.
We can construct a $\DCSACM(1)$ $M$ to accept $L$ as follows.
$M$ reads $a^i$ and stores $a^i$ in the stack.  Then it reads 
$b^j$ and increments the counter by $j$. Finally, $M$ reads
$c^k$ while moving up and down the stack containing $a^i$
and decrementing the counter by 1 every time the stack has moved $i$ cells, to verify that  $k$ is divisible by $i$
and $k/i = j$.  Then $M$ accepts $L$, and $M$ needs
only one 1-reversal counter.
We will see in Proposition \ref{new4} that $L$ cannot be accepted by a $2\DCM(1)$.
\end{example}

The following shows that, in general,  $\NCSACM$s and $\DCSACM$s
are computationally more powerful than
$\NCSA$s and $\DCSA$s, respectively.

\begin{proposition} \label{Grei}
There are languages in $\LL(\DCSACM(1,1)) - (\LL(\NCSA)\cup \LL(\NCM))$.
Hence, $\LL(\DCSA) \subsetneq \LL(\DCSACM(1,1))$, and
$\LL(\NCSA) \subsetneq \LL(\NCSACM(1,1))$.
\end{proposition}
\begin{proof}
Consider the language 
$L = \{(a^n\#)^n ~|~ n \ge 1\}$ from Example \ref{DCSAwitness}.  
$L$ cannot be accepted by
an $\NCSA$; otherwise, $L' = \{a^{n^2} ~|~ n \ge 1\}$ can also be
accepted by an $\NCSA$ (since $\NCSA$ languages are closed under
homomorphism), but it was shown in \cite{CheckingStack} that
$L'$ cannot be accepted by any $\NCSA$. 
However, Example \ref{DCSAwitness} showed that $L$ can
be accepted by a $\DCSACM(1,1)$. Furthermore, $L$ is not semilinear, but $\NCM$ only accepts semilinear languages \cite{Ibarra1978}.
\qed \end{proof}

We now proceed to show that the membership problem for
$\DCSACM$s is decidable.
In view of Lemma \ref{generallambda}, our problem reduces to deciding,
given a $\DCSACM$ $M$, whether it accepts $\lambda$.
For acceptance of $\lambda$, the next lemma
provides a normal form.

\begin{lemma}\label{alwayswrites}
Let $M$ be a $\DCSACM$. We can effectively construct 
a $\DCSACM$ $M'$ such that:
\begin{itemize}
\item all counters of $M'$ are $1$-reversal-bounded and
each must return to zero before accepting,
\item $M'$ always writes on the stack at 
each step during the writing phase,
\item the stack head returns to the left end of the stack before
accepting,
\end{itemize}
whereby $M'$ accepts $\lambda$ if and only if $M$ accepts $\lambda$.
\end{lemma}
\begin{proof} 
It is evident that all counters can be assumed to be $1$-reversal-bounded
as with $\DCM$ \cite{Ibarra1978}, and that each counter can be forced to return
to zero before accepting. Similarly, the checking
stack can be forced to return to the left end before accepting.
We introduce a dummy symbol $\$$ to the stack alphabet
so that if in a step, $M$ does not write on the stack, then $M'$
writes $\$$.  When $M'$ enters the reading phase, $M'$ simulates 
$M$ but ignores (i.e., skips over) the $\$$'s. Then 
$M'$ accepts $\lambda$ if and only if $M$ accepts $\lambda$.
\qed \end{proof}

In view of Lemma \ref{alwayswrites}, we 
may assume that a $\DCSACM$ writes
a symbol at the end of the stack at each step during the
writing phase. This is important for deciding the following problem.
\begin{lemma} \label{infiniteloop}
Let $M$ be a $\DCSACM$ satisfying the assumptions of Lemma
\ref{alwayswrites}.  We can effectively
decide whether or not $M$, on $\lambda$ input,  has an infinite writing
phase (i.e., will keep on writing). 
\end{lemma}
\begin{proof}   Let $s$ be the number of states of $M$.  We construct 
an $\NCM$ $M'$ which, when given an input $w$ over the stack
alphabet of $M$, does the following: simulates 
the computation of $M$ on `stay' transitions while checking that 
$w$ could be written by $M$ on the stack at some point during the computation of the writing phase of $w$, while also verifying that there is a subword  $x$
of $w$ of length $s +1$ such that $x$ was written by $M$
without:\begin{enumerate}
\item incrementing a counter that has so far been at zero, and
\item decrementing  a non-zero counter.
\end{enumerate}
If so, $M'$ accepts $w$. Next, it will be argued that
$L(M')$ is not empty if and only if $M$ has an infinite
writing phase on $\lambda$, and indeed this is decidable since
emptiness for $\NCM$ is decidable \cite{Ibarra1978}.

If $L(M')$ is not empty, 
then there is a sequence of $s+1$ transitions during the writing phase
where no counter during this sequence is increased from zero,
and no counter is decreased. Thus, there must be some state $q$
hit twice by the pigeonhole principle, and the sequence of transitions
between $q$ and itself must repeat indefinitely in $M$.
Thus, $M$ has an infinite writing
phase on $\lambda$ input.

Conversely, assume $M$ has an infinite writing phase.
Then there must be a sequence of $s+1$ transitions where no
counter is decreased, and no counter is increased from zero.
Thus, $L(M')$ must be non-empty.
\qed \end{proof}

From this, decidability of acceptance of $\lambda$ is straightforward.
\begin{lemma}  It is decidable, given a $\DCSACM$ $M$  
satisfying the assumptions of Lemma \ref{alwayswrites},
whether or not $M$ accepts $\lambda$.
\label{lemmaafternormal}
\end{lemma}
\begin{proof}
From Lemma \ref{infiniteloop}, we can decide if $M$ has an
infinite writing phase.  If so, $M$ will not
accept $\lambda$ (as the stack must return to the bottom before accepting).  

If  $M$ does not have an infinite writing phase,  
the (final) word $w$ written in the stack is unique and 
hence has a unique length $d$.  In this case, we 
can simulate faithfully the computation of $M$ (on $\lambda$ input)
and determine $d$.

We then construct a DCM $M_d$, which on 
$\lambda$ input, encodes the stack in the state and simulates $M$.
Thus, $M_d$ needs a buffer of size $d$ to simulate the operation
of the stack, and $M_d$ accepts if and only if $M$ accepts.
The result follows, since 
the membership problem for $\DCM$ is decidable \cite{Ibarra1978}.
\qed \end{proof}

From Lemmas \ref{generallambda}, \ref{alwayswrites}, and \ref{lemmaafternormal}:
\begin{proposition} \label{prop6}
For $r \geq 1$, the membership problem for $r$-head $2\DCSACM$
is decidable.
\end{proposition}

We now give some undecidability results. The proofs will use the 
following result in \cite{Ibarra1978}:

\begin{proposition} \cite{Ibarra1978} \label{hilbert}
It is undecidable, given a $2\DCM(2)$ $M$  
over a letter-bounded language, whether $L(M)$ is empty.  
\end{proposition}


\begin{proposition} The membership problem
for $\NCSACM(2)$ is undecidable.
\label{membershipNCSACM}
\end{proposition}
\begin{proof} Let $M$ be a $2\DCM(2)$ machine over a letter-bounded
language. Construct from $M$
an $\NCSACM$ $M'$ which, on $\lambda$ input (i.e. the input is fixed),  guesses an input  
$w$ to $M$ and  writes it on its stack. Then $M'$ simulates the computation of $M$ by using 
the stack and two reversal-bounded counters and accepts if and
only if $M$ accepts.  Clearly, $M'$ 
accepts $\lambda$ if and only if $L(M)$ is not empty which is undecidable
by Proposition \ref{hilbert}.
\qed \end{proof}

By Propositions \ref{equivalentdecidability} and \ref{membershipNCSACM}, the following is true:
\begin{corollary} \label{dec1}
The emptiness problem for $\DCSACM(2)$ is undecidable.
\end{corollary}

Combining together the results thus far demonstrates that $\NCSACM$ is a model where,
\begin{itemize}
\item the deterministic version has a decidable membership problem,
\item the deterministic version has an undecidable emptiness problem,
\item the nondeterministic version has an undecidable membership problem,
\item the nondeterministic version has an undecidable emptiness problem.
\end{itemize}
Moreover, this is the first (to our knowledge) model where these properties hold.

The next restriction serves to contrast this undecidability result.
Consider an $\NCSACM$ where during the reading phase,
the stack head crosses the boundary of any two adjacent cells
on the stack at most  $d$ times for some given $d \ge 1$.  Call this
machine a $d$-crossing $\NCSACM$.
Then we have:

\begin{proposition} \label{crossing}
It is decidable, given a $d$-crossing $\NCSACM$ $M$, 
whether or not $L(M) = \emptyset$.
\end{proposition}  
\begin{proof}

Define a $d$-crossing $\NTMCM$ to be an nondeterministic Turing machine 
with a one-way
read-only input tape and a $d$-crossing read/write worktape (i.e.,
the worktape head crosses the boundary between any two
adjacent worktape cells at most $d$ times) augmented with 
reversal-bounded counters. Note that a $d$-crossing
$\NCSACM$ can be simulated by a $d$-crossing $\NTMCM$. 
It was shown in \cite{Harju2002278} that
it is decidable, given a $d$-crossing $\NTMCM$ $M$,
whether $L(M) = \emptyset$. The proposition follows.
\qed \end{proof}

Although we have been unable to resolve the open problem
as to whether the emptiness problem is decidable for both $\NCSACM$
and $\DCSACM$ with one reversal-bounded counter, as with membership
for the nondeterministic version, we show they are all equivalent
to an open problem in the literature.
\begin{proposition}
\label{equivalenttoopen}
The following are equivalent:
\begin{enumerate}
\item the emptiness problem  is decidable for $2\NCM(1)$,
\item the emptiness problem is decidable for $\NCSACM(1)$,
\item the emptiness problem is decidable for $\DCSACM(1)$,
\item the membership problem is decidable for $r$-head $2\NCSACM(1)$,
\item it is decidable if $\lambda$ is accepted by a $\NCSACM(1)$.
\end{enumerate}
\end{proposition}
\begin{proof}
The last four properties are equivalent by Proposition
\ref{equivalentdecidability}.

It can be seen that 2) implies 1) because a $\NCSACM(1)$ machine can simulate a $2\NCM(1)$ machine by taking the
input, copying it to the stack, then simulating the $2\NCM(1)$ machine with the two-way stack instead of the two-way input.

Furthermore, it can be seen that 1) implies 5) as follows: given a 
$\NCSACM(1)$ machine $M$, 
assume without loss of generality, that $M$ immediately and
nondeterministically sets the stack and returns to the bottom of
the stack in read-only mode in some special state $q$ before
changing any counter (as it can verify that $M$ would have pushed the
stack contents). Then, build a $2\NCM(1)$ machine $M'$ that on
some input over the stack alphabet, simulates the stack using the input,
and the counter using the counter starting at state $q$. 
Then $L(M')$ is non-empty if and 
only if $\lambda$ is accepted by $M$.
\qed \end{proof}
It is indeed a longstanding open problem as to whether the emptiness problem for $2\NCM(1)$
is decidable \cite{Ibarra1978}. 

Now consider the following three restricted models, with $k$ counters:
For $k \geq 1$, a $\DCSACM(k)$ (or a $\NCSACM(k)$) machine is said to be:
\begin{itemize}
\item {\em no-read/no-counter} if it does not read the checking stack nor use any counter before hitting the right input end-marker,
\item {\em no-read/no-decrease} if it does not read the checking stack nor decrease any counter before hitting the right input end-marker,
\item {\em no-read} if it does not read the checking stack before hitting the right input end-marker.
\end{itemize}
We will consider the families of $\DCSACM(k)$ ($\NCSACM(k)$) machines
satisfying each of these three conditions.

\begin{proposition}
For any $k \geq 1$, every $2\DCM(k)$ machine can be effectively converted to an equivalent no-read/no-decrease $\DCSACM(k)$ machine, and vice-versa.
\label{EquivalenceDCSACMCounter}
\end{proposition}
\begin{proof}
First, a $2\DCM(k)$ machine $M$ can be simulated by a no-read/no-decrease $\DCSACM(k)$ machine $M'$
that first copies the input to the stack, and simulates the input of $M$ using
the checking stack, while simulating the counters faithfully.
Indeed, the checking stack is not read and counters are not
decreased until $M'$ reads the entire input.

Next we will prove the converse. 
Let $M$ be a no-read/no-decrease $\DCSACM(k)$ machine with input alphabet $\Sigma$ and stack alphabet $\Gamma$.

A two-way deterministic gsm, $2{\rm DGSM}$, is a deterministic
generalized sequential machine with a two-way
input (surrounded by end-markers), accepting states, and output.
It is known that if $L$ is a language accepted by a two-way 
$k$-head deterministic machine augmented with some storage/memory structure  
(such as a pushdown, checking stack, $k$ checking stacks, etc.), then
$T^{-1}(L)= \{ x  \mid  T(x) = y,  y \in L\}$ is also accepted by the same type of machine \cite{EngelfrietCheckingStack}.

Let $T$ be $2{\rm DGSM}$ which, on input $x \in \Sigma^*$, first outputs  $x \#$.
Then it moves to the left end-marker and on the second
sweep of input $x$, simulates $M$ and outputs the string $z$ written on the
stack during the writing phase of $M$. Note that $T$ can successfully do this as $M$ generates the checking
stack contents from left-to-right, and does not read the contents during the writing phase; and because 
the counters of $M$ are not decreased
during the writing phase of $M$, the counters can never empty during the writing phase, thereby affecting
the checking stack contents created.
When $T$ reaches 
its right end-marker, it outputs the state $s$ of $M$ at that time, and 
then $T$ enters an accepting state.  Thus, $T(x) = x \# z s$.

Now construct a $2\DCM(k)$ $M'$ which when given a string $x \#z s$,
reads $x$, and while doing so, $M'$ simulates the writing phase of $M$ on $x$ by only changing the counters
as $M$ would do. Then, $M'$ moves to the right and stores the state $s$ in the finite
control.  Then $M'$ simulates the reading phase of $M$ on string $z$ (which only happens after the end of the input has been reached), starting in state $s$
and the current counter contents, and accepts if and only if
$M$ accepts.  

It is straightforward to see that
$T^{-1}(L(M')) = L$, which can therefore be accepted by a $2\DCM(k)$ machine.
\qed \end{proof}

From this, the following is immediate, since emptiness for
$2\DCM(1)$ is known to be decidable \cite{IbarraJiang}.
\begin{corollary} \label{cor2}
The emptiness problem for no-read/no-decrease $\DCSACM(1)$ is decidable.
\end{corollary}

In the first part of the proof of Proposition \ref{EquivalenceDCSACMCounter}, the $\DCSACM(k)$ machine
created from a $2\DCM(k)$ machine was also no-read/no-counter.
Therefore, the following is immediate:
\begin{corollary}
For $k \geq 1$, the family of languages accepted by the following three sets of machines coincide:
\begin{itemize}
\item all no-read/no-decrease $\DCSACM(k)$ machines,
\item all no-read/no-counter $\DCSACM(k)$ machines,
\item $2\DCM(k)$.
\end{itemize}
\end{corollary}

One particularly interesting corollary of this result is the following:
\begin{corollary} \label{cor4}
\begin{enumerate}
\item The family of languages accepted by no-read/no-decrease (respectively no-read/no-counter)
$\DCSACM(1)$ is effectively closed under union, intersection, and complementation.
\item Containment and equivalence are decidable for languages accepted by no-read/no-decrease $\DCSACM(1)$ machines.
\end{enumerate}
\end{corollary}
This follows since this family is equal to $2\DCM(1)$, and these results
hold for $2\DCM(1)$ \cite{IbarraJiang}.
Something particularly noteworthy about closure of languages accepted by no-read/no-decrease $2\DCSACM(1)$
under intersection, is that, the proof does not follow the usual approach for one-way machines. Indeed,
it would be usual to simulate two machines in parallel, each requiring its own counter (and checking stack). But here, only one counter
is needed to establish intersection,
by using a result on two-way machines.
Later, we will show that Corollary \ref{cor4}, part 2 also holds
for no-read $\DCSACM(1)s$.

Also, since emptiness is undecidable for $2\DCM(2)$, even over letter-bounded languages \cite{Ibarra1978}, the following is true:
\begin{corollary}
\label{emptinessNoReadNoCounter}
The emptiness problem for languages accepted by no-read/no-counter $\DCSACM(2)$ is undecidable, even over letter-bounded languages.
\end{corollary}

Turning now to the nondeterministic versions, from the first part of Proposition \ref{EquivalenceDCSACMCounter}, it
is immediate that for any $k \geq 1$, every $2\NCM(k)$ can be effectively converted to an equivalent
no-read/no-decrease $\NCSACM(k)$. But, the converse is not true combining together the following two facts:
\begin{proposition}
\label{contrast}
\begin{enumerate}
\item For every $k \geq 1$, the emptiness problem for languages accepted by $2\NCM(k)$ over a unary alphabet is decidable.
\item The emptiness problem for languages accepted by no-read/no-counter (also for no-read/no-decrease) $\NCSACM(2)$ over a unary alphabet is undecidable.
\end{enumerate}
\end{proposition}
\begin{proof}
The first part  was shown in \cite{IbarraJiang}.
For the second part, it is known that the emptiness problem for $2\DCM(2)$ $M$ (even
over a letter-bounded language) is undecidable by
Proposition \ref{hilbert}.
We construct a no-read/no-counter $\NCSACM(2)$ $M'$ which, on a unary input,
nondeterministically writes some string $w$ on the stack.  Then $M'$
simulates $M$ using $w$.  The result follows since $L(M') = \emptyset$ if and only
if $L(M) = \emptyset$.
\qed \end{proof}

In contrast to part 2 of Proposition \ref{contrast}:
\begin{proposition}
For any $k \geq 1$, the emptiness problem for languages accepted by no-read/no-decrease $\DCSACM(k)$ machines
over a unary alphabet, is decidable.
\end{proposition}
\begin{proof}
If $M$ is a no-read/no-decrease $\DCSACM(k)$ over a unary alphabet, we can effectively construct an
equivalent $2\DCM(k)$ $M$ (over a unary language) from Proposition \ref{EquivalenceDCSACMCounter}. 
The result follows since the emptiness problem for $2\NCM(k)$ over unary languages is decidable \cite{IbarraJiang}.
\qed \end{proof}

Combining these two results yields the following somewhat strange contrast:
\begin{corollary}
\label{emptinessNCSACM}
Over a unary input alphabet and for all $k \geq 2$, 
the emptiness problem for 
no-read/no-counter $\NCSACM(k)s$ is undecidable, but
decidable for no-read/no-counter $\DCSACM(k)s$.
\end{corollary}

As far as we know, this demonstrates the first known 
example of a family of
one-way acceptors where the nondeterministic
version has an undecidable emptiness problem,
but the deterministic version has a decidable emptiness problem.
This presents an interesting contrast to Proposition \ref{equivalentdecidability}, where it was shown that for
complete sets of automata for any store types, the emptiness problem of the deterministic version is decidable if
and only if it is decidable for the nondeterministic version. However, the set of unary 
no-read/no-counter $\NCSACM(k)$ machines can be seen to not be a {\bf complete} set of machines, as a complete set of machines
contains every possible machine involving a store type. This
includes those machines that read input letters while performing
read instructions on the checking stack. 
And indeed, to prove the equivalence of 1) and 2)
in Proposition \ref{equivalentdecidability}, the deterministic machine
created reads a letter for every transition applied, which can produce machines that are not of the restriction
no-read/no-counter.

\label{pagething}

With only one counter, 
decidability of the emptiness problem for no-read/no-decrease
$\NCSACM(1)$, and for no-read/no-counter $\NCSACM(1)$ can
be shown to be equivalent to all problems
listed in Proposition \ref{equivalenttoopen}.
This is because 2) of Proposition \ref{equivalenttoopen} implies each immediately, and each implies 1) of Proposition
\ref{equivalenttoopen}, as a 
$2\NCM(1)$ machine $M$ can be converted to a no-read/no-decrease, or no-read/no-counter $\NCSACM(1)$
machine where the input is copied to the stack, and then
the $2\NCM(1)$ machine simulated.

Therefore, it is open as to whether the emptiness problem for no-read/no-decrease (or no-read/no-counter) $\NCSACM(1)$ 
is decidable, as this is equivalent to the
emptiness problem for $2\NCM(1)$. 
One might again suspect that
decidability of emptiness for no-read/no-decrease $\DCSACM(1)$ implies
decidability of emptiness for no-read/no-decrease $\NCSACM(1)$ by Proposition \ref{equivalentdecidability}. However, it is again important to note
that Proposition \ref{equivalentdecidability} only applies
to complete sets of machines,
including those machines that read input letters while performing
read instructions on the checking stack, again violating the `no-read/no-decrease' condition.

Even though it is open as to whether the emptiness problem is decidable
for no-read/no-decrease $\NCSACM(1)s$, we have the
following result, which contrasts Corollary \ref{cor4}, part 2:

\begin{proposition}
The universe problem is undecidable for no-read/no-counter
$\NCSACM(1)$s.  (Thus, containment and equivalence are undecidable.)
\end{proposition}
\begin{proof}
It is known that the universe problem for a one-way nondeterministic
1-reversal-bounded one-counter automaton $M$ is undecidable \cite{Baker1974}.  Clearly,
we can construct a  no-read/no-counter $\NCSACM(1)$ $M'$ to simulate $M$.
\qed \end{proof}

In the definition of a no-read/no-decrease $\DCSACM$, we imposed the condition that the counters can only decrement
when the input head reaches the end-marker.  
Consider the weaker condition no-read, i.e., the only requirement
is that the machine can only enter the stack when the input head reaches the end-marker, but there
is no constraint on the reversal-bounded counters. It is
an interesting open question about whether no-read $\DCSACM(k)$ languages are also equivalent to a $2\DCM(k)$ (we conjecture
that they are equivalent). However, the following stronger version of Corollary \ref{cor2} can be proven.
\begin{proposition} \label{prop12}
The emptiness problem is decidable
for no-read $\DCSACM(1)$s.
\end{proposition}
\begin{proof}
Let $M$ be a no-read $\DCSACM(1)$. 
Let $T = \{t_1, \ldots, t_m\}$ be symbols in bijective correspondence with transitions of $M$ that can occur in the writing phase. Then, build a $2\DCM(1)$ machine $M'$ that, on input $w$ over $T$, reads $w$ while changing states as $w$ does, and changing the counter as the transitions do. Let $q$ be the state where the last transition symbol ends. Then, at the end of the input, $M'$ simulates the reading phase of $M$ starting in $q$ by scanning $w$, and interpreting a letter $t$ of $w$ as being the stack letter written by $t$ in $M$ (while always skipping over a letter $t$ if $t$ does not write to the stack in $M$). Then $L(M')$ is empty if and only if $L(M)$ is empty.
\qed \end{proof}

We can further strengthen Proposition \ref{prop12} somewhat.  Define a restricted
no-read $\NCSACM(1)$ to be a no-read $\NCSACM(1)$ which is only
nondeterministic during the writing phase.  Then the proof of
Proposition \ref{prop12} applies to the following, as the sequence
of transition symbols used in the proof can be simulated deterministically:

\begin{corollary}
The emptiness problem is decidable for languages accepted by
restricted no-read $\NCSACM(1)$ machines.
\end{corollary}

While we are unable to show that the intersection of
two no-read $\DCSACM(1)$ languages is a no-read $\DCSACM(1)$ 
language, we can prove:

\begin{proposition} \label{new1}
It is decidable,  given  two no-read $\DCSACM(1)$s
$M_1$ and $M_2$, whether $L(M_1) \cap L(M_2) = \emptyset$.
\end{proposition}
\begin{proof}
Let $M_1$ and $M_2$ be no-read $\DCSACM(1)$ over input alphabet $\Sigma$. Let $T_i$ be symbols in bijective correspondence with transitions of $M_i$ that can occur in the writing phase, for each $i \in\{1,2\}$. Let $T'$ be the set of all pairs of symbols $(r,s)$, where $r$ is a transition of $M_1$, $s$ is a transition of $M_2$, and where both $r$ and $s$ read the same input letter of $\Sigma$. Let $T''$ be all those symbols $(r,\$)$ where $r$ is a transition of $M_1$ that stays on the input, and let $T'''$ be all those symbols $(\$,s)$ where $s$ is a transition of $M_2$ that stays on the input.

Build a $2\DCM(1)$ machine $M'$ operating over alphabet $T' \cup T'' \cup T'''$. On input $w$, $M'$ verifies that the first component changes states as $M_1$ does (skipping over any \$ symbol) and that if a stay transition is read, the next letter has a first component on the same input letter, and changing the counter as $M_1$ does. Let $q$ be the state where the last transition symbol ends. Then, at the end of the input, $M'$ simulates the reading phase of $M_1$ starting in $q$ by scanning $w$, and interpreting a letter $t \neq \$$ in the first component of $w$ as being the stack letter written by $t$ in $M$, and skipping over $\$$ or any $t$ that does not write to the stack. After completion, then $M'$ does the same thing with $M_2$ using the second component. Notice that the alphabet is structured such that a transition of $M_1$ on a letter $a \in \Sigma$ is used exactly when a transition of $M_2$ using $a \in \Sigma$ is used, since $M_1$ and $M_2$ are both no-read (so their entire input is used before the reading phases starts). For example, a word $w = (s_1,r_1) (s_2, \$) (s_3,\$) (\$, r_2) (s_4,r_3)$ implies $s_1$ reads the same input letter in $M_1$ as does $r_1$ in $M_2$, similarly with $s_4$ and $r_3$, $s_2$ and $s_3$ are stay transitions in $M_1$, and $r_2$ is a stay transition in $M_2$.
Hence, $L(M')$ is empty if and only if $L(M_1) \cap L(M_2)$ is empty.
\qed \end{proof}

One can show that no-read  $\DCSACM(1)$ languages
are effectively closed under complementation.  Thus, 
from Proposition \ref{new1}:

\begin{corollary} \label{new2}
The containment and equivalence problems are decidable for
no-read $\DCSACM(1)$s.
\end{corollary}

No-read $\DCSACM(1)$ is indeed quite a large family for which emptiness, equality, and containment are decidable.
The proof of Proposition \ref{new1} also applies to the following:

\begin{proposition} \label{new3}
It is decidable,  given  two restricted no-read $\NCSACM(1)$s
$M_1$ and $M_2$, whether $L(M_1) \cap L(M_2) = \emptyset$.
\end{proposition}

Finally, consider the general model $\DCSACM(1)$ (i.e., unrestricted).
While it is open whether no-read $\DCSACM(1)$ is equivalent to $2\DCM(1)$,
we can prove:

\begin{proposition} \label{new4}
$\LL(2\DCM(1)) \subsetneq \LL(\DCSACM(1))$.
\end{proposition}
\begin{proof}
It is obvious that any $2\DCM(1)$ can be simulated by 
a $\DCSACM(1)$ (in fact by a no-read/no-counter $\DCSACM(1)$).
Now let $L = \{a^i b^j c^k ~|~ i, j \ge 1, k = i \cdot j \}$.
We can construct a $\DCSACM(1)$ $M$ to accept $L$ by Example 
\ref{2DCM1witness}.
However, it was shown in
\cite{Gurari1982} that $L$ cannot be accepted by a $2\DCM(1)$
by a proof that shows that if $L$ can be accepted by a $2\DCM(1)$,
then one can use the decidability of the emptiness problem
for $2\DCM(1)$s to show that Hilbert's Tenth Problem is decidable.
\qed \end{proof}

\section{Multiple Checking-Stacks with Reversal-Bounded Counters}

In this section, we will study
deterministic and nondeterministic $k$-checking-stack
machines. These are defined by using 
multiple checking stack stores.
Implied from this definition is that each stack has
a ``writing phase'' followed by a ``reading phase'', but
these phases are independent of each letter for each stack.

A $k$-stack $\DCSA$ ($\NCSA$ respectively) is the deterministic (nondeterministic) version of this type of
machine. The two-way versions (with input end-markers)
are called $k$-stack $2\DCSA$ and $k$-stack $2\NCSA$, respectively.
These $k$-stack models can also be augmented with reversal-bounded
counters and are called $k$-stack $\DCSACM$, $k$-stack $\NCSACM$,
$k$-stack $2\DCSACM$, and $k$-stack $2\NCSACM$. 

Consider a $k$-stack $\DCSACM$ $M$.
By Lemma \ref{generallambda}, for the membership problem, we need only
investigate whether $\lambda$ is accepted. 
Also, as in Lemma \ref{alwayswrites},
we may assume that each stack pushes a symbol at each move
during its writing phase, and that all 
counters are $1$-reversal-bounded.

We say that $M$ has an infinite writing phase (on $\lambda$ input)
if no stack enters a reading phase. Thus, all stacks will keep
on writing a symbol at each step.  If $M$ has a finite writing phase,
then directly before a first such stack enters its reading phase,
all the stacks would have written strings of the same length.

\begin{lemma} \label{lem13}
Let $k \ge 1$ and $M$ be a $(k+1)$-stack $\DCSACM$ $M$
satisfying the assumption of Lemma \ref{alwayswrites}.
\begin{enumerate}
\item
We can determine if $M$ has an infinite writing phase.
If so, $M$ does not accept $\lambda$. 
\item
If $M$ has a finite writing phase, we can construct
a $k$-stack $\DCSACM$ $M''$ satisfying the assumption of Lemma \ref{alwayswrites} such that
$M''$ accepts $\lambda$ if and only if $M$ accepts 
$\lambda$.
\end{enumerate}
\end{lemma}
\begin{proof}
Let $M$ have $s$ states and stack alphabets
$\Gamma_1, \ldots, \Gamma_{k+1}$ for the $k+1$ stacks.
Let $\Gamma = \{[a_1, \ldots,$ $a_{k+1}] ~|~ a_i \in \Gamma_i, 1 \leq i \leq k+1 \}$.
By assumption, each stack of $M$ writes a symbol during its
writing phase.  

We can determine if $M$ has a finite writing phase
as follows: As 
in Lemma \ref{infiniteloop}, we construct 
an $\NCM$ $M'$ which, when given an input
$w \in \Gamma^*$, does the following: simulates the computation
of $M$ on `stay' transitions such that
the input $w$ was written by $M$ (in a component-wise fashion on each checking stack) and there is a subword  $x$
of $w$ of length $s+1$ such that the subword was written by $M$
without:\begin{enumerate}
\item incrementing a counter that has so far been at zero, and
\item decrementing  a non-zero counter.
\end{enumerate}
If so, $M'$ accepts $w$.
So we need only check if $L(M')$ is not empty, which is decidable
since emptiness is decidable for $\NCM$ \cite{Ibarra1978}. Then,
$M$ does not accept $\lambda$ if and only if
$M$ has an infinite writing phase, and if and only if
$L(M')$ is not empty, which is decidable.

If $L(M')$ is empty, we then simulate $M$ faithfully
to determine the unique word $w \in \Gamma^*$ and its
length $d$ just before the reading phase of at least one
of the stacks, say $S_i$, is  entered.
Note that by construction, no stack entered its
stack earlier.

We then construct a $k$-stack $\DCSACM$ $M''$
which, on 
$\lambda$ input, encodes the operation of stack $S_i$
in the state and simulates $M$ (also converted into satisfying the assumptions of Lemma \ref{alwayswrites}).
Thus, $M''$ needs a buffer of size $d$ to simulate the operation
of stack $S_i$.  $M''$ accepts if and only if $M$ accepts,
and has one less stack than $M$.
\qed \end{proof}

Notice that $M''$ has fewer stacks than $M$. Then, from Proposition \ref{prop6} (the result for a single stack) 
and using Lemma \ref{lem13} recursively:
\begin{proposition}  The membership problem for $k$-stack
$\DCSACM$s is decidable.
\end{proposition}
Then, by Lemma \ref{generallambda}:
\begin{corollary} The membership problem for $r$-head $k$-stack $2\DCSACM$
is decidable.
\label{fulldecidable}
\end{corollary}

This is one of the most general machine models known with a decidable membership problem. Although space complexity classes of Turing machines are also very general, the membership problem for both deterministic and nondeterministic Turing machines satisfying some space complexity function are both decidable. However, for $\NCSACM$s, membership is undecidable but is decidable for deterministic machines. Moreover, unlike space-bounded Turing machines, $r$-head  $k$-stack $2\DCSACM$s
do not have a space restriction on their stacks.

\section{Conclusions}
We introduced several variants of checking stack automata and showed
the difference between the deterministic and nondeterministic models
with respect to the decidability of the membership and emptiness
problems. The main decision problems are summarized in Table \ref{summary}. We believe the contrasting results obtained are the first
of its kind.  An interesting open question is the status of the emptiness
problem for nondeterministic checking stack automata  augmented
with one reversal-bounded counter which can only read the stack
and decrease the counter at the end of the input. As shown in the
paper,  this problem is equivalent to a
long-standing open problem of
whether emptiness for two-way nondeterministic finite automata
augmented with one reversal-bounded counter is decidable.
Furthermore, we investigated possible scenarios that can occur when augmenting 
a machine model accepting non-semilinear languages with reversal-bounded counters. This contrasts known results on models accepting only semilinear languages.

\begin{table}\begin{center}
\begin{tabular}{|l|c|c||c|c|c|}
\hline
\textbf{nondeterministic class} & membership  & emptiness &\textbf{deterministic class} & membership & emptiness  
\\\hline 
 \hline 
 $\NCSACM(1)$  & $?$ & $?$ & \DCSACM(1) & $\checkmark$ & $?$ \\ 
	& Prop \ref{equivalenttoopen} & Prop \ref{equivalenttoopen} & & Prop \ref{prop6} & Prop \ref{equivalenttoopen}  \\\hline
$\NCSACM(k), k \geq 2$ & $\times$ & $\times$ & $\DCSACM(k),k \geq 2$ & $\checkmark$ & $\times$  \\
		& Prop \ref{membershipNCSACM} & Prop \ref{membershipNCSACM} & & Prop \ref{prop6} & Cor \ref{dec1}  \\\hline
$r$-head $k$-stack & $\times$ & $\times$ & $r$-head $k$-stack & $\checkmark$ & $\times$   \\
$2\NCSACM$ & Prop \ref{membershipNCSACM} & Prop \ref{membershipNCSACM} & $2\DCSACM$  &  Cor \ref{fulldecidable} & Cor \ref{dec1} \\\hline
no-read/no-decrease  &	$?$ & $?$ & no-read/no-decrease & $\checkmark$ & $\checkmark$ \\
$\NCSACM(1)$ & pg \pageref{pagething} & pg \pageref{pagething} & $\DCSACM(1)$ & Prop \ref{prop6} & Cor \ref{cor2}  \\\hline
no-read/no-counter  &	$\times$ & $\times$ & no-read/no-counter & $\checkmark$ & $\times$ \\
$\NCSACM(2)$ & Prop \ref{contrast} & Prop \ref{contrast} & $\DCSACM(2)$ & Prop \ref{prop6} & Cor \ref{emptinessNoReadNoCounter}  \\\hline
\end{tabular} \end{center}
\caption{Summary of decision problem results, with the left half being nondeterministic machines, and the right half being the corresponding deterministic machines. The decision problems are listed on the columns. A $\checkmark$ is placed when the problem is decidable for the class of machines, a $\times$ is placed when the property is undecidable, and a $?$ is placed when the problem is still open. In all cases where it is open, the decidability is equivalent to the open problem of whether emptiness is decidable for $2\NCM(1)$. The theorem proving each is listed.}
\label{summary}
\end{table}

\section*{Acknowledgements}
We thank the Editor and the referees for their expeditious
handling of our paper and, in particular, the referees for their
comments that improved the presentation of our results.

\bibliography{bounded}{}
\bibliographystyle{elsarticle-num}

\end{document}